\documentclass[10pt,a4paper]{article}
\usepackage[a4paper,hmargin=3cm,vmargin=4cm]{geometry}
\usepackage[utf8x]{inputenc}
\usepackage{amsmath,amsthm,amsfonts,amssymb}
\usepackage{authblk}
\usepackage[bookmarksnumbered=true]{hyperref}
\usepackage{bbm}
\usepackage{bigints}
\usepackage{setspace}
\usepackage{tikz-cd}
\usepackage{stackrel}

\usepackage{color}
\usepackage{appendix}

\usepackage{graphicx}

\numberwithin{equation}{section}

\newtheorem{theorem}{Theorem}[section]

\newtheorem{lemma}[theorem]{Lemma}
\newtheorem{proposition}[theorem]{Proposition}
\theoremstyle{definition}

\def\wt{\widetilde}

\def\eps{\varepsilon}
\def\N{\mathbb{N}}

\def\R{\mathbb{R}}
\let\e=\varepsilon

\let\.=\cdot
\let\0=\emptyset

\def\tr{\text{\rm tr}}

\newcommand{\esssup}{\mathop{\rm ess{\,}sup}}

\def\square{\hbox{$\sqcap\kern-7pt\sqcup$}}

\def\be{\begin{equation}}
\def\ee{\end{equation}}
\def\bea{\begin{eqnarray}}
\def\eea{\end{eqnarray}}

\def\o{\omega}
\def\ˆ{^{•}}

\DeclareMathOperator*{\argmin}{argmin}

\title{From the Hartree equation to the Vlasov-Poisson system:\\ strong convergence for a class of mixed states}

\author{Chiara Saffirio}

\def\adresse{
\begin{description}

\item[C. Saffirio:] Department of Mathematics and Computer Science, \\ University of Basel, Spiegelgasse 1, CH-4051 Basel, Switzerland\\
E-mail: \texttt{chiara.saffirio@unibas.ch}
\end{description}
}

\date{\today}

\begin{document}

\maketitle

\begin{abstract}
We consider the evolution of $N$ fermions interacting through a Coulomb or gravitational potential in the mean-field limit as governed by the nonlinear Hartree equation with Coulomb or gravitational interaction. In the limit of large $N$, in our setting corresponding to the semiclassical limit, we study the convergence in trace norm towards the classical Vlasov-Poisson equation for a special class of mixed quasi-free states. 
\end{abstract}

\section{Introduction}

In this paper we shall focus on the derivation of the three-dimensional Vlasov-Poisson system
\be\label{eq:VPS}
\left\{\begin{array}{ll}
 \partial_t f(t,x,v)+v\cdot\nabla_x f(t,x,v)+E\cdot\nabla_v f(t,x,v)=0\,,& \\\\
 E(t,x)=\left(\nabla\frac{\gamma}{|\cdot|}*\varrho\right)(t,x)\,,&\quad (t,x,v)\in\R_+\times\R^3\times\R^3\\\\
 \varrho(t,x)=\int f(t,x,v)\,dv\,,&
\end{array}
\right.
\ee
from the quantum $N$-body dynamics in a joint mean-field and semiclassical regime. More precisely, we will address the semicalssical limit from the mean-field quantum description given by the Hartree equation towards the Vlasov-Poisson equation. In Eq. \eqref{eq:VPS}, $\gamma=-1$ or $\gamma=1$ if the interaction is respectively gravitational or Coulombian. The system \eqref{eq:VPS} is an effective equation, whose unknown $f:\R_+\times\R^3\times\R^3\to \R$ is a time dependent function on the phase space modelling the probability density of particles in a plasma under the effect of a self-induced field $E$, dependent on the spatial density $\varrho:\R_+\times\R^3\to\R_+$. From a physical viewpoint, the attractive case ($\gamma=-1$) describes the motion of galaxy clusters under the gravitational field with many applications in astrophysics.  The repulsive case ($\gamma=1$) represents the evolution of charged particles in presence of their self-consistent electric field and it is used in plasma physics or in semi-conductor devices.  In this context,  the self-induced field $E(t,x)$ is a conservative force, hence there exists a real-valued function of time and space $U(t,x)$ such that $E=\nabla_x U$ which satisfies the Poisson equation $-\Delta_x U = \varrho_t$. More precisely, Eq. \eqref{eq:VPS} can be rewritten as a Vlasov equation coupled with a Poisson equation, whence the name Vlasov-Poisson system.

Many authors have been investigating the problem of deriving the Vlasov-Poisson system \eqref{eq:VPS} from the many-body quantum dynamics and its understanding is nowadays rather complete, at least in the case of particles at positive temperature, without statistics or obeying the Bose statistics. The situation is much less clear when dealing with fermions. The aim of this paper is to provide a better understanding of the derivation of the Vlasov-Poisson system from the dynamics of $N$ particles obeying the Fermi statistics in a coupled mean-field and semiclassical scaling. More precisely, it focuses on the vertical arrow of the diagram in \eqref{fig:diagram} when the interaction among particles is given by the Coulombian or by the gravitational potential.
\smallskip

{\it Many-body dynamics.} We consider a system of $N$ fermions in dimension $d=3$ interacting through a Coulomb potential (or  a gravitational potential). A state of the system is described by a wave function $\psi_N\in L^2_{a}(\R^{3N})$, the space of square integrable functions which are antisymmetric in the exchange of particles. This property reflects the Pauli exclusion principle, which states that two particles cannot occupy the same quantum state. 

The Hamiltonian associated with a system of $N$ fermions interacting through a Coulomb or gravitational potential which acts on $L^2_{a}(\R^{3N})$ is given by
\be\label{eq:H}
H_N=\sum_{j=1}^N-\Delta_{q_j}+\gamma\sum_{i< j}^N\frac{1}{|q_i-q_j|}
\ee
where $q_j\in\R^3$ denotes the position of the $j$-th particle, the potential energy is positive in the repulsive case $\gamma=+1$, i.e. for the Coulomb potential, and negative $\gamma=-1$ in the case of gravitational interaction. We shall focus now on the repulsive case, though the exact same considerations can be done for the attractive case by changing the sign in front of the potential energy. \\
Given a quantum state $\psi_N\in L^{2}_{a}(\R^{3N})$, its time evolution is driven by the $N$-body Schr\"{o}dinger equation
\be\label{eq:MBS}
\left\{
\begin{array}{l}
i\partial_\tau\psi_{N,\tau}(q_1,\dots,q_N)=H_N\psi_{N,\tau}(q_1,\dots,q_N)\,,\\\\
\psi_{N,0}(q_1,\dots,q_N)=\psi_{N}(q_1,\dots,q_N)\,.
\end{array}
\right.
\ee
Though the Cauchy problem \eqref{eq:MBS} admits the unique solution $\psi_{N,t}=e^{-itH_N}\psi_N$, the number of particles $N$ is typically so large that we cannot extract any useful information out of it. We therefore look for an effective equation which approximates in a suitable regime the $N$-body dynamics.
\smallskip

{\it Fermionic mean-field regime.} In this paper we are interested in the mean-field regime. To see how the mean-field scaling arises in this context, we consider an electrically neutral atom made of $N$ electrons (fermions) and a nucleus fixed at the origin. The associated Hamiltonian is given by
\be\label{eq:H-atom}
H_N^{\rm atom}=\sum_{j=1}^N\left[-\Delta_{q_j}-\frac{N}{|q_j|}\right]+\sum_{i<j}^N\frac{1}{|q_i-q_j|}\,.
\ee
From Thomas-Fermi theory (see for instance \cite{Lieb, LSi}), we know that the $N$ electrons localise at distance $O(N^{-\frac{1}{3}})$ from the fixed nucleus. It is therefore natural to scale the position variables by a factor $N^{-\frac{1}{3}}$. This leads to rewrite the Hamilton operator \eqref{eq:H-atom} in the scaled variables as
\be\label{eq:scaled-atom}
H_N^{\rm atom}=\sum_{j=1}^N\left[-N^{\frac{2}{3}}\Delta_{x_j}-\frac{\ N^{\frac{4}{3}}}{|x_j|}\right]+N^{\frac{1}{3}}\sum_{i<j}^N\frac{1}{|x_i-x_j|}\,,
\ee
with $x_j=N^{\frac{1}{3}}q_j$. We introduce the positive  parameter 
\be\label{eq:eps}
\e_N=N^{-\frac{1}{3}}
\ee
to rewrite \eqref{eq:scaled-atom} as
\be
H_N^{\rm atom}=N^{\frac{4}{3}}\left[ \sum_{j=1}^N\left(-\e_N^2\Delta_{x_j}-\frac{1}{|x_j|}\right)+\frac{1}{N}\sum_{i<j}^N\frac{1}{|x_i-x_j|}\right]\,.
\ee
We now scale the time variable in order to absorb the factor $N^{\frac{4}{3}}$ in front of the Hamilton operator, that is equivalent to look at larger time scales in a new time variable $t$, and  we readily see that such a space-time scaling leads to the $N$-body Schr\"{o}dinger equation
\be
i\e_N\partial_t\psi_{N,t}=\left[ \sum_{j=1}^N\left(-\e_N^2\Delta_{x_j}-\frac{1}{|x_j|}\right)+\frac{1}{N}\sum_{i<j}^N\frac{1}{|x_i-x_j|}\right]\psi_{N,t}
\ee
where the factor $N^{-1}$ in front of the interaction is typical of the mean-field regime. We also observe that such a mean-field scaling comes coupled with a semiclassical limit since $\e_N$ defined in \eqref{eq:eps} plays the role of the Plank constant $\hbar$.
This regime is very different from the bosonic one, in which the mean-field approximation and the semiclassical limit are not coupled. \\
For the rest of the paper we will discard the subscript $N$ in $\e_N$ and denote it simply by $\e$ for the  sake of readibility. 
\smallskip

{\it State of  art.} The first derivation of the Vlasov equation from many-body quantum dynamics was obtained in the 80s by Narnhofer and Sewell \cite{NS} for interaction potentials $V\in C^\omega(\R^3)$ and extended to $V\in C^2(\R^3)$ by Spohn in \cite{Spohn81}. These results establish convergence of the dynamics, but no information on the rate of convergence is provided.    
An analogous result has been obtained in \cite{GMP} by analysing the dynamics of factored WKB states 
and combining the mean-field and the semiclassical limit.

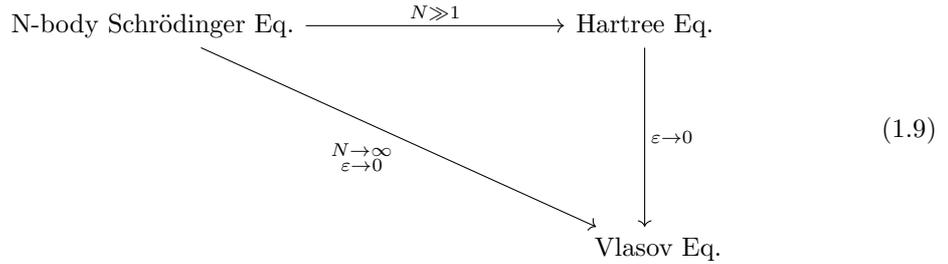
\be
    \begin{tikzcd}\label{fig:diagram} 
  \mbox{N-body Schr\"{o}dinger Eq.} \rar [rd,"\substack{N\to\infty\\ \e\to 0}"'] \rar [r,rightarrow, "\quad\quad N\gg 1\quad\quad"] 
    &[7em] \mbox{Hartree Eq.} \arrow[d,rightarrow,"\e\to 0"]  \\ [5em]
&\quad \mbox{Vlasov Eq.} \end{tikzcd}
\ee


A different approach consists in considering the Hartree equation as a bridge between the $N$-body Schr\"odinger dynamics and the Vlasov equation, as pictured in  \eqref{fig:diagram}.  
The Hartree equation reads
\be\label{eq:Hartree}
i\e\partial_t\ \omega_{N,t}=\left[-\e^2\Delta+V*\varrho_t\,,\,\omega_{N,t}\right]
\ee
where $\omega_{N,t}$ is a nonnegative trace class fermionic operator over $L^2(\R^3)$, and for two operators $A$ and $B$ the standard notation $[A,B]$ stands for the commutator $AB-BA$. More precisely, we say that an operator $\omega_{N,t}$ is fermionic if the bound
\be\label{eq:fermions}
0\leq\omega_{N,t}\leq 1
\ee
hold. Notice moreover that if we assume \eqref{eq:fermions} at time $t=0$, the equation \eqref{eq:Hartree} propagates such a bound for positive times.

Looking at the Hartree equation \eqref{eq:H}, one observes that its solution $\omega_{N,t}$ is still $N$ dependent and one expects it to approach a solution to the Vlasov equation as $N\to\infty$. 
Figure \eqref{fig:diagram} above describes two ways of deriving the Vlasov equation from the many-body dynamics: either one performs a direct limit $N=\e^{-3}\to \infty$, or one first observes that for $N$ large but fixed the many-body Schr\"{o}dinger equation is approximated by the Hartree equation and then performs the semiclassical limit $\e\to 0$ recovering the Vlasov equation
\[\left\{\begin{array}{l}
\partial_t W_{t}+v\cdot\nabla_x W_{t}+(V*\varrho_t)\cdot\nabla_v W_{t}=0\,, \\\\
\varrho_t(x)=\int W_{t}(x,v)\,dv\,.
\end{array}
\right.
\]

To compare these objects, namely an operator $\omega_{N,t}$ on $L^2(\R^3)$ (solution to the Hartree equation \eqref{eq:Hartree}) and a function $W_{N,t}$ defined on the phase space $\R^3\times\R^3$ (solution to the Vlasov system \eqref{eq:VPS}), we introduce two standard tools in semiclassical analysis: the Wigner transform of an operator $\omega_{N,t}$ defined as
\be\label{eq:wigner}
W_{N,t}(x,v)=\left(\frac{\e}{2\pi}\right)^3\int \omega_{N,t}\left(x+\e\frac{y}{2};x-\e\frac{y}{2}\right)\,e^{-iv\cdot y}dy\,,
\ee
and its inverse, known as Weyl quantization, given by 
\be\label{eq:weyl}
\omega_{N,t}(x;y)=N\int W_{N,t}\left(\frac{x+y}{2},v\right)\,e^{iv\cdot\frac{(x-y)}{\e}}\,dv\,.
\ee
We recall that the Wigner transform of a fermionic operator is not always a probability density on the phase space, as in general it is not positive. This issue can be fixed by using the Husimi transform (see \cite{BPSS} for further discussions on this point).

The horizontal line in figure \eqref{fig:diagram} in the regime under consideration has been investigated in \cite{ESY, BPS13}, where it has been shown that for regular interaction potentials the Hartree equation is a good approximation for the many-body Schr\"{o}dinger evolution when considering zero temperature states enjoying a semiclassical structure. More precisely, in \cite{ESY} the convergence of the Schr\"{o}dinger dynamics towards a solution of the Hartree equation in the sense of reduced density matrices has been proved for short times and without a control on the rate of convergence. In \cite{BPS13} the authors prove that such an approximation holds for time intervals of order one and provide explicit estimates on the speed of convergence in terms of the number of particles $N$. This latter has been extended to positive temperature states in \cite{BJPSS} and to fermions with semi-relativistic dispersion relation in \cite{BPS-rel}. For singular potentials of the form $V(x)=|x|^{-\alpha}$, for $\alpha\in(0,1]$, it has been shown in \cite{PRSS, S18} that the Hartree dynamics is still a good approximation of the many-body one, at least for a very special class of initial data, namely translation invariant states. \\
Different regimes have been considered in \cite{BBPPT, BGGM1, BGGM2, BDGM, FK, P, PP}. More precisely,  states confined in a volume of order $O(N)$ have been studied in \cite{BBPPT, P, PP}, while a regime in which the potential energy is sub-leading with respect to the kinetic one has been considered in \cite{BGGM1,BGGM2,BDGM,FK}.   
 
Using \eqref{eq:wigner}, in \cite{LionsPaul, GIMS, MM, FLP} the vertical line in figure \eqref{fig:diagram} has been investigated. The authors prove convergence in weak sense towards the solutions of the Vlasov equation. The analysis in \cite{LionsPaul, FLP} includes the Coulomb potential, but does not provide explicit bounds on the convergence rate, which are fundamental for applications. Indeed, in all relevant situations, the number of particles $N$ is large, but finite. An explicit control on the convergence rate therefore allows to determine how large the system (i.e. the number of particles $N$) should be in order for the Vlasov equation to be a meaningful approximation.  

The paper \cite{APPP} was the first of a long list of references in which this aspect has been tackled. Indeed, in \cite{APPP} Athanassoulis, Paul, Pezzotti and Pulvirenti obtain the convergence of the Wigner transform of a solution to the Hartree equation towards a solution to the Vlasov equation in Hilbert-Schmidt norm with a relative rate $N^{-\frac{2}{21}}$ for $V\in H^1(\R^3)$. In \cite{PP} and \cite{AKN1,AKN2} an expansion of the solution of the Hartree equation in powers of $\e$ has been provided.\\ 
More recently, in the same spirit of  \cite{APPP}, assumptions on the potentials has been relaxed first to interactions $V$ such that $\nabla V\in {\rm Lip}(\R^3)$ \cite{BPSS} and then to inverse power law potentials $V(x)~=~\gamma |x|^{-\alpha}$, for $\alpha\in (0,1/2)$ and $\gamma=\pm 1$ (cf. \cite{S18}). A key ingredient is to consider the problem from the perspective of the Weyl quantization, instead of the one  of the Wigner transform usually adopted in the previous literature. In the same regime, convergence of minimizers of the $N$-particle energy to the mean-field energy has been proven in \cite{FLS, LMT}, respectively for zero and positive temperature. \\
A different approach has been recently proposed by \cite{GolsePaul1} with the introduction of a new pseudo-distance which is reminiscent of the Monge-Kantorovich distance for probability measures in the setting of classical mechanics. Under appropriate conditions on the initial states, such a pseudo-distance metrizes weak convergence. For potentials $V$ such that the force satisfies $\nabla V\in{\rm Lip}(\R^3)$, Golse and Paul prove that the Vlasov equation is a good approximation for the Hartree dynamics when considering a special class of bosonic states, defined through T\"{o}plitz operators in \cite{GolsePaul1} and they are able to consider projection operators in \cite{GolsePaul2} by introducing the notion of quantum empirical measure. In \cite{GolsePaulPulvirenti} the convergence result in \cite{GolsePaul1} is proven to be uniform in the Planck constant. On the same line, \cite{GolsePaul1} has been extended  by Lafl\`eche in \cite{Lafleche} to the case of singular interaction potentials, here included the Coulomb case.

{In this paper we are interested in giving a strong notion of convergence of the Hartree dynamics to the Vlasov-Poisson one for particles obeying the Fermi statistics, exhibiting explicit control on the convergence rate and thus extending and complementing previous results by \cite{LionsPaul, FLP, Lafleche}.
\smallskip

{\it Notations, assumptions and main result. } To state precisely the main result of this paper, we introduce the following notations. For $s\in\N$, let $H^s(\R^3\times\R^3)$ denote the Hilbert space of real-valued functions $f$ on the phase space $\R^3\times\R^3$ with finite norm 
\begin{equation*}
\|f\|_{H^s}^2:=\sum_{|\beta|\leq s}\iint_{\R^3\times\R^3} |\nabla^\beta f(x,v)|^2\,dx\,dv\,,
\end{equation*}
where $\beta$ is a multi-index and $\nabla^\beta$ acts on both $x$ and $v$.  For $s,m\in\N$, let $H^s_m(\R^3\times\R^3)$ denote the Sobolev space $H^s$ weighted with $(1+x^2+v^2)^m$ and define its associated norm
\begin{equation*}
\|f\|^2_{H_m^s}:=\sum_{|\beta|\leq s}\iint_{\R^3\times\R^3} (1+x^2+v^2)^m|\nabla^\beta f(x,v)|^2\,dx\,dv\,.
\end{equation*}
We denote the $m$-th velocity moment associated to a function $f\in L^1(\R^3\times\R^3)$ by
\be\label{eq:moments-0}
\mathcal{M}_m(f)=\iint_{\R^3\times\R^3} |v|^m f(x,v)\,dx\,dv\,.
\ee
Moreover, we denote by  $\varrho_{|[x,{\omega}_{N}]|}$ the function associated with the integral kernel of the operator $|[x,\omega_N]|$
\be\label{eq:density-sc}
\varrho_{|[x,{\omega}_N]|}(x):=|[x,{\omega}_N]|(x;x)\,,\quad \mbox{ for } x\in\R^3\,.
\ee

Moreover, throughout the paper we look at the situation in which a smooth solution to the Vlasov-Poisson equation exists and is unique. This is always the case when the following assumptions on the initial datum $W_N$ are satisfied:
\begin{itemize}
\item[$i)$] $W_N\in L^1\cap L^{\infty}(\R^3\times\R^3)$ and  $\mathcal{M}_m(W_N)<\infty$ for all $m<m_0$, with $m_0>6$. 
\item[$ii)$] For all $R,\,T>0$, 
\be\label{eq:uniq-cond}
\begin{split}
\esssup_{y,\,w}\{|\nabla^k W_N|(y+vt,w)\,:\,|y-x|\leq R,&\,|w-v|\leq R\}\\
&\in\ L^\infty((0,T)\times\R_x^3;L^1(\R_v^3)\cap L^2(\R^3_v))
\end{split}
\ee
for $k=0,1, \dots, 5$.
\item[$iii)$] There exists $C>0$ independent of $N$ such that, for $k=0,\dots,6$, $\|W_N\|_{H_4^k}\leq C\,.$
%
%
\end{itemize} 

In the same spirit of \cite{BPSS, S19}, the main result of this paper reads

\begin{theorem}\label{thm:tr-HS}
Let $\e=\e_N$ be defined as in \eqref{eq:eps}. Let $\omega_N$ be a sequence of fermionic operators on $L^2(\R^3)$, $0\leq \omega_N\leq 1$, with $\tr\ \omega_N=N$ and with Wigner transform $W_N$ satisfying $i),\,ii),\,iii)$.
Let $\omega_{N,t}$ denote the solution of the time-dependent Hartree equation \eqref{eq:Hartree} with initial data $\omega_{N,0}=\omega_N$ and let $\widetilde{W}_{N,t}$ be the solution of the Vlasov-Poisson system  \eqref{eq:VPS} with initial data $\widetilde{W}_{N,0}=W_N$. Let $\widetilde{\omega}_{N,t}$ denote the Weyl quantization of $\widetilde{W}_{N,t}$ and assume that there exist a time $T>0$, a number $p>5$ and a positive constant $C$ such that
\be\label{eq:sc-assumptions}
\sup_{t\in [0,T]}\sum_{i=1}^3\left[ \|\varrho_{|[x_i,\widetilde{\omega}_{N,t}]|}\|_1+\|\varrho_{|[x_i,\widetilde{\omega}_{N,t}}\|_p \right]\leq CN\e\,.
\ee 
Then
\be\label{eq:tr-norm}
\tr\ |\omega_{N,t}-\widetilde{\omega}_{N,t}|\leq C_t\,N\,\e\left[1+\sum_{i=1}^4 \e^i\sup_N\|W_N\|_{H_4^{i+2}}\right]\,.
\ee
where $C_t$ is a constant depending only on the time $t$ and on $\|W_N\|_{H^2_4}$.
\end{theorem}

\noindent We comment on the meaning of the Theorem \ref{thm:tr-HS} above and discuss the assumptions. 
\begin{itemize}
\item We notice that because of the normalization $\tr\ \omega_{N,t}=N$, Eq. \eqref{eq:tr-norm} proves that $\omega_{N,t}$ and $\widetilde{\omega}_{N,t}$ are close as $N$ is large. Indeed, the trace norm of $\omega_{N,t}-\widetilde{\omega}_{N,t}$ is smaller than the trace of $\omega_{N,t}$ and $\widetilde{\omega}_{N,t}$. 

\item We observe that assumptions $ii)$ and $iii)$ restrict our analysis to the case of fermionic mixed states, i.e. equilibrium states at positive temperature. Indeed, at zero temperature the states of a system at equilibrium can be approximated by Slater determinants, whose Wigner transforms are in general not smooth. To support this statement, we analyse the following prototype example. Consider a system of $N$ non interacting fermions in a box $\Lambda$ of volume of order one with periodic boundary conditions. This choice makes the system translation invariant, thus guaranteeing that fermions are uniformly distributed in the box with density of particles equal to $N$. We observe that the fact that fermions are not interacting drastically simplifies the picture. Indeed, in this framework the ground state of the system is exactly a Slater determinant whose Wigner transform can be computed explicitly:
\be\label{eq:wigner-pure}
W_N(x,v)=\dfrac{1}{N}\mathbf{1}_{\{(x,v)\in\Lambda\times\R^3\ :\ |v|\leq C\, N^{1/3}\}}
\ee
Therefore, in the case of free fermions in a box of order one with periodic boundary conditions, the Wigner transform of the ground state of the system behaves as the characteristic function of a set and surely does not satisfy the hypotheses of Theorem \ref{thm:tr-HS}. \\
If we consider the interacting picture in which fermions interact through a mean-field potential $V$ and the system is confined by an external potential $V^{\rm ext}$, then the density of particles is not equal to $N$ anymore, but one still expects the Wigner transform associated with the approximation of the ground state of the system to be of the form \eqref{eq:wigner-pure}, where the uniform density $N$ is replaced by a local function of $x$, say $\varrho_{\rm TF}(x)$, that is the minimiser of the Thomas-Fermi energy functional, i.e. 
$$
\varrho_{\rm TF}(x):=\argmin_{\substack{\varrho\in L^1\cap L^{\frac{5}{3}}\\ \|\varrho\|_1=N}}\left\{ \frac{3}{5}c_{\rm TF}\int \varrho(x)^{5/3}dx+\int V^{\rm ext}(x)\varrho(x)dx+\dfrac{1}{2}\iint V(x-y)\varrho(x)\varrho(y)dxdy \right\} 
$$
 where $c_{\rm TF}$ is the Thomas-Fermi constant. 
We therefore conclude that the assumptions of Theorem \ref{thm:tr-HS} are expected to hold for states describing systems of $N$ particles in equilibrium at positive temperature, but do not include pure states, i.e. states at zero temperature (cf. \cite{BPSS} for details). 

\item We comment on the assumptions. Hypotheses $i)$ and $ii)$ (for $k=0,1$) where proven in \cite{LP} to guarantee existence, uniqueness and propagation of moments of a solution to the Vlasov-Poisson system \eqref{eq:VPS}. Hypotheses $ii)$ (for $k=2,\dots,5$) and $iii)$ are crucial to ensure regularity of the solution in $H_4^6$. 
The bounds \eqref{eq:sc-assumptions}  are assumed to hold true at positive time. This is a severe restriction of our result. The quantity \eqref{eq:density-sc} has already played a central role in \cite{BPS13,PRSS,S18}. We recall that the assumption
\be\label{eq:sc-tr}
\|\varrho_{|[x,\omega_{N}]|}\|_{L^1}=\tr\ |[x,\omega_{N}]|\leq CN\e
\ee
is equivalent to ask for initial states enjoying a semiclassical structure. More precisely, the Hartree equation is expected to be a good approximation for the many-body Schr\"{o}dinger dynamics if the kernel of the initial data $\omega_{N}(x;y)$ is concentrated on the diagonal and decays when $|x-y|\gg \e$. Thus, as pointed out in \cite{BPS13}, the kernel of $\omega_N$ should be of the form
\be\label{eq:sc-structure}
\omega_{N}(x;y)\simeq\frac{1}{\e^3}\varphi\left(\frac{x-y}{\e}\right)\varrho\left(\frac{x+y}{2}\right)\,,
\ee 
where $\varrho$ represents the density of particles in space and $\varphi$ fixes the momentum distribution. In particular, \eqref{eq:sc-structure} is compatible with \eqref{eq:sc-tr}.
 See Chapter 6 in\cite{BPS-springer} for a detailed explanation and Section 5 in \cite{BPS13} for the propagation in time of \eqref{eq:sc-tr} in the case of smooth interaction potentials. \\
 To deal with the singularity of the Coulomb potential, we need also to control $L^p$ norms of $\varrho_{|[x,\widetilde{\omega}_{N,t}]|}$ for $p>5$, that means to require more structure on the operator $|[x,\widetilde{\omega}_{N,t}]|$. At the moment, when $p>1$, we do not know how to prove that if $\|\varrho_{|[x,\widetilde{\omega}_N]|}\|_{L^p}\leq CN\e$, this bound is propagated in time by the Vlasov-Poisson evolution. Theorem \ref{thm:tr-HS} is therefore a result conditioned to the uniform in time bounds \eqref{eq:sc-assumptions}.
 Nevertheless, there is one peculiar situation in which it is possible to verify assumption \eqref{eq:sc-assumptions} holds true and it will be discussed in Section \ref{sect:steady-states}.
 
\item  Theorem \ref{thm:tr-HS} is stated for the Hartree dynamics, but it is actually true also for the Hartree-Fock equation
\be\label{eq:HF}
\left\{\begin{array}{l}
i\e\partial_t\ \omega_{N,t}=\left[-\e^2\Delta+V*\varrho_t-X_t\,,\,\omega_{N,t}\right]\,,\\\\
X_t(x;y)=\dfrac{1}{N}V(x-y)\omega_N(x;y)\,,
\end{array}\right.
\ee
for in this setting the exchange term turns out to be sub-leading (cf. \cite{BPS13,PRSS}).  

\end{itemize}


{\it Strategy.} We present here the strategy of the proof in an informal way. We proceed as in \cite{BPSS} by performing a comparison between solutions to the Hartree equation and the Vlasov-Poisson system at the level of operators. This means to consider the Vlasov-Poisson system in its Weyl quantized form (see Eq. \eqref{eq:VP-weyl}). More precisely, we consider a sequence of fermionic operators $\omega_N$, i.e. operators such that $0\leq \omega_N\leq 1$, and look at their evolution first accordingly to the Hartree equation and then according to the Weyl quantized Vlasov-Poisson equation. We denote the solution to the Cauchy problem associated to the Hartree equation with initial data $\omega_N$ by $\omega_{N,t}$, whereas the solution of the Weyl quantized Vlasov-Poisson system with initial data $\omega_N$ is denoted by $\widetilde{\omega}_{N,t}$. We recall that such a solution exists and it is unique under the assumptions of our theorem, due to the result by Lions and Perthame \cite{LP} and its extension to signed measure (see \cite{BPSS} and \cite{S18}). We therefore compare $\omega_{N,t}$ and $\widetilde{\omega}_{N,t}$ looking for a Gr\"{o}nwall type inequality on the trace norm of the operator $\omega_{N,t}-\widetilde{\omega}_{N,t}$. The fist difficulty to cope with is a bound on the kinetic energy associated with the difference $\omega_{N,t}-\widetilde{\omega}_{N,t}$. To overcome this issue, in the same spirit of \cite{BPSS} we introduce a reference frame through a unitary dynamics. This suffices to tackle the problem since in this new reference frame the kinetic energy term 
\[
[\ -\e^2\,\Delta\ ,\ \omega_{N,t}-\widetilde{\omega}_{N,t}\ ]
\]
does not appear anymore (see Lemma \ref{lemma:traceH-V}).  With respect to \cite{BPSS}, where interactions with two bounded derivatives have been treated, the new difficulty here is to deal with the singularity at zero of the Coulomb potential. To face this more fundamental point, we make use of an expression for the Coulomb potential introduced by Fefferman and de la Llave (see Lemma \ref{lemma:FDL}). To be more precise, we employ a smooth version of it (see Eq. \eqref{eq:FDL-smooth}):
\[
\frac{1}{|x-y|}=C\int_0^\infty\int_{\R^3} \frac{1}{r^5}\chi_{(r,z)}(x)\chi_{(r,z)}(y)\,dz\,dr
\]
where $\chi_{(r,z)}(\cdot)$ is a smooth function depending on the distance $|\cdot - z|$ and varying on a scale $r$.
The most important implication of such a rewriting of the Coulomb potential as an integral over
all possible spheres of radius $r\geq 0$ consists in isolating the singularity at zero from all the other terms appearing in the Gr\"{o}nwall like type inequality (most of them produce errors which are estimated in Proposition \ref{prop:error-term}). Hence, the key idea is to cancel part of the interaction by estimating the trace norm of the commutator $[\chi_{(r,z)},\widetilde{\omega}_{N,t}]$, where $\chi_{(r,z)}$ acts as a multiplication operator. In order to absorb the singularity at $r=0$, we need the bound on $\tr\ |[\chi_{(r,z)},\widetilde{\omega}_{N,t}]|$ to be sharp in the $r$ variable (see Lemma \ref{lemma:tr1}). 
{Such a sharp bound forces us to look at quantities (see Eq. \eqref{eq:density-sc}) that are in general not known to be bounded in terms of the assumptions on the initial data. This is the reason why we need to restrict our analysis to a special class of initial data (see Section \ref{sect:steady-states}).  }
\medskip

{\it Plan of the paper.} In Section \ref{sect:lemmas} we present some preliminary estimates which will be used in Section \ref{sect:proof}, where the proof of Theorem \ref{thm:tr-HS} is presented; we conclude with Section \ref{sect:steady-states}, where examples of initial states verifying the assumptions of Theorem \ref{thm:tr-HS} are analysed, namely translation invariant states for the Vlasov-Poisson system and steady states for the attractive Vlasov-Poisson system. Hence, Theorem \ref{thm:tr-HS} shows that the Hartree evolution for non translation invariant states and for non stationary states can be approximated by translation invariant states, or respectively steady states in the attractive case, of the Vlasov-Poisson system, thus showing that the Hartree dynamics leaves the state of the system approximately invariant. 

\section{Auxiliary Lemmas and Propositions}\label{sect:lemmas}

We start by giving a handier expression for the trace norm of the difference of a solution to the Hartree equation \eqref{eq:H-coulomb} and the Weyl transform of the solution to the Vlasov-Poisson system \eqref{eq:VPS}. 
\begin{lemma}\label{lemma:traceH-V}
Let $\omega_N$ be a sequence of fermionic operators on $L^2(\R^3)$, $0\leq \omega_N\leq 1$, with $\tr\ \omega_N=N$ and denote by $W_N$ its Wigner transform. 
For each $N\in\N$, let $\omega_{N,t}$ be the solution of the time-dependent Hartree equation with Coulomb interaction
\be\label{eq:H-coulomb}
i\,\e\,\partial_t\omega_{N,t}=\left[-\e^2\Delta+\frac{1}{|\cdot|}*\varrho_t\,,\,\omega_{N,t}\right]
\ee
 with initial data $\omega_{N,0}=\omega_N$ and let $\widetilde{W}_{N,t}$ be the solution of the Vlasov-Poisson system  \eqref{eq:VPS} with initial data $\widetilde{W}_{N,0}=W_N$. Moreover, let $\widetilde{\omega}_{N,t}$ denote the Weyl quantization of $\widetilde{W}_{N,t}$. Then the following estimate holds true
\be\label{eq:traceH-V}
\tr \, |\o_{N,t}-\wt\o_{N,t}|\leq \frac{1}{\e}\int_0^t \tr \, \left|\left[\frac{1}{|\cdot|}*(\varrho_s-\wt\varrho_s),\wt\o_{N,s}\right]\right|\,ds + \frac{1}{\e}\int_0^t \tr \, |B_s|\,ds \,,
\ee
where for every $t\geq 0$, $B_t$ is the operator associated with the kernel
\be\label{eq:B}
B_t (x;y) = \left[ \left(\frac{1}{|\cdot|}*\wt{\varrho}_t\right) (x) - \left(\frac{1}{|\cdot|}*\wt{\varrho}_t\right)(y) - \nabla \left(\frac{1}{|\cdot|}*\wt{\varrho}_t\right) \left(\frac{x+y}{2} \right) \cdot (x-y) \right] \wt{\omega}_{N,t} (x;y)
\ee
for all $x,y\in\R^3$.
\end{lemma}
\begin{proof}
We perform the Weyl transform of the Vlasov-Poisson system \eqref{eq:VPS}
\begin{equation*} \left\{
\begin{array}{l}
\partial_t \widetilde{W}_{N,t} + v\cdot\nabla_x \widetilde{W}_{N,t} + E\cdot\nabla_v \widetilde{W}_{N,t} = 0\,,\\\\
E(t,x)=\nabla\left(\frac{\gamma}{|  \cdot  |}*\widetilde{\varrho}_t\right)(x)\,,\\\\
\widetilde{\varrho}_t(x)=\int \widetilde{W}_{N,t}(x,v)\,dv\,,
\end{array}
\right.
\end{equation*}
 and we obtain
\be\label{eq:VP-weyl}
i\,\e\,\partial_t\widetilde{\omega}_{N,t}=[-\e^2\Delta\,,\,\widetilde{\omega}_{N,t}]+A_t
\ee
where $\widetilde{\omega}_{N,t}$ is the Weyl transform of $\widetilde{W}_{N,t}$ and $A_t$ is the operator associated with the kernel
\begin{equation*}
A_t(x;y)=\nabla\left(\frac{1}{|\cdot|}*\widetilde{\varrho}_t\right)\left(\frac{x+y}{2}\right)\cdot(x-y)\,\widetilde{\omega}_{N,t}(x;y)\,.
\end{equation*}
Since we are interested in finding an expression for the difference of the operators $\omega_{N,t}$ and $\widetilde{\omega}_{N,t}$, we look for a Gr\"{o}nwall type estimate and compute the quantity $i\e\partial_t(\omega_{N,t}-\widetilde{\omega}_{N,t})$. To cope with the kinetic terms, we introduce a fictitious unitary dynamics given by the two-parameter group of unitary transformations $\mathcal{U}(t;s)$. Its time evolution is given by
\be\label{eq:unitary}
i\,\e\,\partial_t\,\mathcal{U}(t;s)=h_H(t)\,\mathcal{U}(t;s)\,,
\ee
where $h_H=-\e^2\Delta+\frac{1}{|\cdot|}*\varrho_t$ is the Hartree Hamiltonian.  

We then conjugate the operator $(\omega_{N,t}-\widetilde{\omega}_{N,t})$ by $\mathcal{U}(t;s)$. We observe that such a choice makes $\omega_{N,t}$ play the role of a reference frame, thus we get
\be
\begin{split}
i\,\e\,\partial_t\,\mathcal{U}^*(t;0)\,&(\o_{N,t}-\wt\o_{N,t})\,\mathcal{U}(t;0) \\ = \; &-\mathcal{U}^*(t;0)\,[h_H (t),\o_{N,t}-\wt\o_{N,t}]\,\mathcal{U} (t;0)\\
& +\mathcal{U}^*(t;0)\,([h_H(t),\o_{N,t}]-[-\e^2\,\Delta,\wt\o_{N,t}] - A_t)\,\mathcal{U}(t;0)\\
= \; &\mathcal{U}^*(t;0)\,\left(\left[\frac{1}{|\cdot|}*\varrho_t,\wt\o_{N,t}\right]-A_t\right)\,\mathcal{U}(t;0)\\
= \; &\mathcal{U}^*(t;0)\,\left(\left[\frac{1}{|\cdot|}*(\varrho_t-\wt\varrho_t),\wt\o_{N,t}\right]+ B_t\right)\,\mathcal{U}(t;0)
\end{split}
\ee
where $B_t$ denotes the operator with the integral kernel defined in \eqref{eq:B}. 

Since at time $t=0$ $\omega_{N,0} = \wt{\omega}_{N,0} = \omega_N$, integration in time gives 
\be\label{eq:omega-omegatilde}
\begin{split}
\mathcal{U}^*(t;0)\,(\o_{N,t}-\wt\o_{N,t})\,\mathcal{U}(t;0) = \;& \frac{1}{i\e}\int_0^t \mathcal{U}^*(t;s)\,\left[\frac{1}{|\cdot|}*(\varrho_s-\wt\varrho_s),\wt\o_{N,s}\right]\,\mathcal{U}(t;s)\,ds \\+&\frac{1}{i\e}\int_0^t \mathcal{U}^*(t;s)\,B_s\,\mathcal{U}(t;s)\,ds\,.
\end{split}
\ee
Taking the trace norm in \eqref{eq:omega-omegatilde}, we obtain 
\be\label{eq:trace-norm}
\tr \, |\o_{N,t}-\wt\o_{N,t}|\leq \frac{1}{\e}\int_0^t \tr \, \left|\left[\frac{1}{|\cdot|}*(\varrho_s-\wt\varrho_s),\wt\o_{N,s}\right]\right|\,ds + \frac{1}{\e}\int_0^t \tr \, |B_s|\,ds 
\ee
as desired.
\end{proof}

We will estimate the two terms in the right-hand side of (\ref{eq:trace-norm}) separately, and conclude by applying Gronwall's lemma.
The key idea is to rewrite the Coulomb interaction as an integral over all possible spheres of radius $r\geq 0$, that is the content of the following Lemma first used in \cite{FDLL} by Fefferman and de la Llave to prove stability of matter in the relativistic case. 
\begin{lemma}\label{lemma:FDL}
For every $x\in\R^3$ and $r\geq 0$, let $\mathbf{1}_{\{|x-z|\leq r\}}$ be the characteristic function of the sphere $\{z\in\R^3\,:\,|x-z|\leq r \}$,
then
\be\label{eq:FDL-coulomb}
\frac{1}{|x-y|}=\frac{1}{\pi}\int_0^\infty \int_{\R^3} \frac{1}{r^5}\,\mathbf{1}_{\{|x-z|\leq r\}}\,\mathbf{1}_{\{|y-z|\leq r\}}\,dz\,dr\,. 
\ee
\end{lemma}
The proof follows by direct inspection.
Moreover, an elementary modification of the above Lemma shows that if one replaces the characteristic function $\mathbf{1}_{\{|x-y|\leq r\}}$ with a smooth version of it, e.g. $\chi_{(r,x)}(y):=\exp (-|x-y|^2/r^2)$, Eq. \eqref{eq:FDL-coulomb} is still valid, provided the numerical constant in front of the integrals is appropriately modified. \\
From now on, we denote 
$$\chi_{(r,x)}(y):=\exp (-|x-y|^2/r^2)$$ 
thus \eqref{eq:FDL-coulomb} reads 
\be\label{eq:FDL-smooth}
\frac{1}{|x-y|}=\frac{4}{\pi^2}\int_0^\infty \int_{\R^3} \frac{1}{r^5}\,\chi_{(r,z)}(x)\,\chi_{(r,z)}(y)\,dz\,dr\,. 
\ee
\begin{lemma}\label{lemma:dominant}
For every $x\in\R^3$ and $r\geq 0$, let $\chi_{(r,x)}(y):=\exp (-|x-y|^2/r^2)$. Then
\be\label{eq:dom-term}
\tr\ \left|\left[ \frac{1}{|\cdot|}*(\varrho_s-\widetilde{\varrho}_s)\,,\,\widetilde{\omega}_{N,s} \right]\right|\leq C\int_0^\infty\frac{1}{r^5}\iint |\varrho_s(y)-\widetilde{\varrho}_s(y)|\,\chi_{(r,z)}(y)\,\tr\ |[\chi_{(r,z)}\,,\,\widetilde{\omega}_{N,s}]|\,dz\,dy\,dr\,.
\ee
\end{lemma}
\begin{proof}
The identity \eqref{eq:FDL-smooth} allows to rewrite the convolution on the l.h.s. of \eqref{eq:dom-term}, for every $x\in\R^3$, as
\begin{equation*}
\frac{1}{|\cdot|}*(\varrho_s-\widetilde{\varrho}_s)(x)=\frac{4}{\pi^2}\int_0^\infty\iint \frac{1}{r^5}\,\chi_{(r,y)}(x)\,\chi_{(r,z)}(y)\,(\varrho_s(y)-\widetilde{\varrho}_s(y))\,dz\,dy\,dr\,.
\end{equation*} 
Therefore, for every $x,x'\in\R^3$, we obtain the following expression for the kernel of the commutator in the l.h.s. of \eqref{eq:dom-term}
\be
\begin{split}
&\left[\frac{1}{|\cdot|}*(\varrho_s-\widetilde{\varrho}_s)\,,\,\widetilde{\omega}_{N,s}\right](x;x')\\ 
&\quad\quad =
\frac{4}{\pi^2}\int_0^\infty \iint \frac{1}{r^5}\,(\varrho_s(y)-\widetilde{\varrho}_s(y))\,\chi_{(r,z)}(y)\,[\chi_{(r,z)}\,,\,\tilde{\omega}_{N,s}](x;x')\,dz\,dy\,dr\,.
\end{split}
\ee
Taking the trace norm of the operator associated with the kernels in the above expression, the bound \eqref{eq:dom-term} holds.
\end{proof}

The following Lemma provides a key estimate to deal with the Coulomb singularity at zero. The proof can be found in \cite{PRSS}, but we report it here for completeness. 
\begin{lemma}[Lemma 3.1 in \cite{PRSS}]
\label{lemma:tr1}
Let $\chi_{(r,z)} (x) := \exp (-|x-z|^2/r^2)$ and, given $T>0$, assume $[x_i,\widetilde{\omega}_{N,t}]$ to be a trace class operator for all $t\in[0,T]$. Then, for all $0<\delta<1/2$ there exists $C>0$ such that the following bound holds point-wise
\be\label{eq:tr1} 
\tr\ \left|[\chi_{(r,z)} , \widetilde{\omega}_{N,t} ] \right| 
\leq C \, r^{\frac{3}{2} - 3\delta} \sum_{i=1}^3 \| \varrho_{|[x_i, \widetilde{\omega}_{N,t}]|} \|_1^{\frac{1}{6}+\delta} \left( \varrho^*_{|[x_i, \widetilde{\omega}_{N,t}]|} (z) \right)^{\frac{5}{6} - \delta}\,,
\ee
where $\varrho^*_{|[x_i,\widetilde{\omega}_{N,t}]|}$ denotes the Hardy-Littlewood maximal function of $\varrho_{|[x_i,\widetilde{\omega}_{N,t}]|}$, defined by 
\be\label{eq:max-def} 
\varrho^*_{|[x_i , \widetilde{\omega}_{N,t}]|} (z) = \sup_{B : z \in B} \frac{1}{|B|} \int_B \varrho_{|[x_i , \widetilde{\omega}_{N,t}]|} (x)\,dx
\ee
where the supremum is taken over all spheres $B$ containing the point $z\in\R^3$, and $\varrho_{|[x_i,\widetilde{\omega}_{N,t}]|}$ is defined in \eqref{eq:density-sc}. 
\end{lemma}
\begin{proof}
We consider the commutator $[\chi_{(r,z)}\,,\,\widetilde{\omega}_{N,t}]$ and we write it as
\[
[\chi_{(r,z)}\,,\,\widetilde{\omega}_{N,t}]=\sum_{j=1}^3\mathcal{I}_j+\mathcal{J}_j\,,
\]
where $\mathcal{I}_j$ and $\mathcal{J}_j$ are defined as follows
\[
\begin{split}
&\mathcal{I}_j=-\int_0^1\chi_{(r/\sqrt{s},z)}(x)\,\dfrac{(x-z)_j}{r^2}\,[x_j\,,\,\widetilde{\omega}_{N,t}]\,\chi_{(r/\sqrt{1-s},z)}(x)\,ds\\
&\mathcal{J}_j=-\int_0^1\chi_{(r/\sqrt{s},z)}(x)\,[x_j\,,\,\widetilde{\omega}_{N,t}]\,\dfrac{(x-z)_j}{r^2}\,\chi_{(r/\sqrt{1-s},z)}(x)\,ds
\end{split}
\]
for $j=1,2,3$. As the two terms can be treat in the same way, we focus on $\mathcal{I}_j$.
By assumption, $[x_j,\widetilde{\omega}_{N,t}]$ is a trace class operator and therefore it has a spectral decomposition. Let $\{f_k\}_k$ be an orthonormal system in $L^2(\R^3)$ and $\{\alpha_k\}_k$ be the associated eigenvalues, where $\alpha_k\in\R$ for all $k$. Then
\be\label{eq:spectral-dec}
[x_j\,,\,\widetilde{\omega}_{N,t}]=i\sum_k \alpha_k |\,f_k\rangle\langle f_k\,|
\ee
Eq. \eqref{eq:spectral-dec} and the definition of trace norm then leads to
\[
\begin{split}
\tr\ |\mathcal{I}_j|&\leq  \dfrac{1}{r}\sum_k |\alpha_k|\int_0^1 \dfrac{1}{\sqrt{s}}\tr\ \left|\,\left|\chi_{(r/\sqrt{s},z)}(x)\dfrac{\sqrt{s}|x-z|}{r}\,f_k\right\rangle\left\langle \chi_{(r/\sqrt{1-s},z)}(x)\,f_k\,\right|\,\right| \,ds\\
&\leq \dfrac{1}{r}\int_0^1 \dfrac{1}{\sqrt{s}}\left(\sum_k |\alpha_k|\left\|\chi_{(r/\sqrt{s},z)}(x)\dfrac{\sqrt{s}|x-z|}{r}f_k\right\|_2^2\right)^{\frac{1}{2}}\left(\sum_k|\alpha_k|\left\|\chi_{(r/\sqrt{1-s},z)}(x)f_k\right\|_2^2\right)^{\frac{1}{2}}
\end{split}
\]
where in the last line we used Cauchy-Schwarz inequality.\\
By using the definition of the kernel of the multiplication operator $\chi_{(r/\sqrt{1-s},z)}(x)$ and again the spectral decomposition \eqref{eq:spectral-dec}, we get the following bounds: on the one hand we have
\be\label{eq:bound-max-funct}
\begin{split}
\sum_k|\alpha_k|\left\|\chi_{(r/\sqrt{1-s},z)}(x)\,f_k \right\|_2^2&=\int_0^1 \int_{B(z,\sqrt{r^2\log(1/t)/2(1-s)})}\varrho_{|[x_j\,,\,\widetilde{\omega}_{N,t}]|}(x)\,dx\,dt\\
&\leq \dfrac{C\,r^{3}}{(1-s)^{3/2}}\varrho^*_{|[x_j,\widetilde{\omega}_{N,t}]|}(z)\,,
\end{split}
\ee
where $C$ is a positive constant, $B(z,\sqrt{r^2\log(1/t)/2(1-s)})$ is the ball centred at $z$ with radius $\sqrt{r^2\log(1/t)/2(1-s)}$ and $\varrho^*_{|[x_j,\widetilde{\omega}_{N,t}]|}$ is the Hardy-Littlewood maximal function of $\varrho_{|[x_j,\widetilde{\omega}_{N,t}]|}$ defined in \eqref{eq:max-def}.\\
On the other hand we have
\be\label{eq:tr1-L1}
\sum_k |\alpha_k|\left\| \chi_{(r/\sqrt{1-s},z)}(x) f_k \right\|_2^2\leq C\sum_k|\alpha_k|\\
=\tr\ |[x_j,\widetilde{\omega}_{N,t}]|=C\|\varrho_{|[x_j,\widetilde{\omega}_{N,t}]|}\|_{L^1}\,.
\ee
Interpolating between \eqref{eq:bound-max-funct} and \eqref{eq:tr1-L1}, we obtain
\[
\sum_k|\alpha_k|\left\| \chi_{(r/\sqrt{s},z)}(x)\frac{\sqrt{s}|x-z|}{r} f_k\right\|_2^2\leq C\dfrac{r^{3\theta}\|\varrho_{|[x_j,\widetilde{\omega}_{N,t}]|}\|_{L^1}^{1-\theta}}{s^{\frac{3}{2}\theta}}\left(\varrho^*_{|[x_j,\widetilde{\omega}_{N,t}]|}(z)\right)^{\theta}
\]
that, with the choice $\theta=\frac{2}{3}-2\delta$ yields the desired bound.
\end{proof}

\begin{proposition}\label{prop:error-term}
Let $B_t$ be the operator associated with the kernel \eqref{eq:B}. Then, there exists a constant $C>0$ depending on $\|\widetilde{W}_{N,t}\|_{H^2_4}$, $\|\widetilde{\rho}_t\|_{L^1}$ and $\|\nabla^2\widetilde{\rho}_t\|_{L^\infty}$ such that
\be
\tr\ |B_t|\leq C\,N\,\e^2\,\left(1+\sum_{k=1}^4 \e^k\|\widetilde{W}_{N,t}\|_{H_4^{k+2}}\right)\,.
\ee
\end{proposition}
Before giving the proof, we remark that the objects on which the constant $C$ depends on are bounded by standard regularity theory for the Vlasov-Poisson system (cf. for instance \cite{Golse}).
\begin{proof}
To bound the trace norm of $B_t$ we introduce the identity operator 
\begin{equation*}
\mathbf{1}=(1-\eps^2\Delta)^{-1} (1+x^2)^{-1}(1+x^2) (1-\eps^2 \Delta).
\end{equation*} 
By applying Cauchy-Schwarz inequality we have
\be\label{eq:tr-B} 
\tr \, |B_t| \leq \| (1-\eps^2\Delta)^{-1} (1+x^2)^{-1} \|_{\rm HS} \, \| (1+x^2) (1-\eps^2 \Delta) B_t \|_{\rm HS}
\ee
We notice that for some $C>0$ the following bound holds
\begin{equation*}
\| (1-\eps^2\Delta)^{-1} (1+x^2)^{-1} \|_{\rm HS} \leq C \sqrt{N} 
\end{equation*}
where we have used the explicit form of the kernel of the operator $(1-\e^2\Delta)^{-1}$ and the fact that $\e^3=N$.\\
We denote  by $U_t$ the convolution of the interaction with the spatial density at time $t$
\be\label{eq:convolution}
U_t:= \frac{1}{|\cdot|}* \widetilde{\varrho}_t.
\ee
We introduce the notation $$\widetilde{B} := (1-\eps^2 \Delta) B_t$$ and observe that the kernel of $\widetilde{B}$ reads 
\be\label{eq:termsB}
\widetilde{B} (x;x'):=\sum_{j=1}^7 \widetilde{B}_j (x;x')  
\ee
where
{\small
\begin{align*}
\widetilde{B}_1 &(x;x')\\
 = &N \left[U_t(x) - U_t(x') - \nabla U_t\left(\dfrac{x+x'}{2}\right)\cdot(x-x')\right]\int \widetilde{W}_{N,t}\left(\dfrac{x+x'}{2},v\right)e^{i\,v\cdot\frac{(x-x')}{\e}}dv; \displaybreak[0]  \\
\widetilde{B}_2 &(x;x')\\
 = &- N\e^2\left[\Delta U_t(x) - \dfrac{1}{4}\Delta \nabla U_t\left(\dfrac{x+x'}{2}\right)\cdot(x-x') - \dfrac{1}{2} \Delta U_t \left(\dfrac{x+x'}{2}\right)\right] \int \widetilde{W}_{N,t}\left(\dfrac{x+x'}{2},v\right)e^{i\,v\cdot\frac{(x-x')}{\e}}dv; \displaybreak[0]\\
\widetilde{B}_3 &(x;x')\\
 = &- \frac{N\e^2}{4} \left[U_t(x) - U_t(x') - \nabla U_t\left(\dfrac{x+x'}{2}\right)\cdot(x-x')\right] \int (\Delta_1 \widetilde{W}_{N,t}) \left(\dfrac{x+x'}{2},v\right)e^{i\,v\cdot\frac{(x-x')}{\e}}dv; \displaybreak[0]\\
\widetilde{B}_4 &(x;x')\\
 = N &\left[U_t(x) - U_t(x') - \nabla U_t\left(\dfrac{x+x'}{2}\right)\cdot(x-x')\right] \int \widetilde{W}_{N,t}\left(\dfrac{x+x'}{2},v\right) v^2 e^{i\,v\cdot \frac{(x-x')}{\e}}dv; \displaybreak[0]\\
\widetilde{B}_5 &(x;x')\\
 = &- \dfrac{N\e^2}{2} \left[\nabla U_t(x) - \dfrac{1}{2} \nabla^2 U_t\left(\dfrac{x+x'}{2}\right) (x-x') - \nabla U_t\left(\dfrac{x+x'}{2}\right)\right] \int (\nabla_1 \widetilde{W}_{N,t}) \left(\dfrac{x+x'}{2},v\right) e^{i\,v\cdot \frac{(x-x')}{\e}}dv; \displaybreak[0]\\
\widetilde{B}_6 &(x;x')\\
 = &- N\e \left[\nabla U_t(x) - \dfrac{1}{2} \nabla^2 U_t\left(\dfrac{x+x'}{2}\right)(x-x') - \nabla U_t\left(\dfrac{x+x'}{2}\right)\right] \int \widetilde{W}_{N,t}\left(\dfrac{x+x'}{2},v\right) v e^{i\,v\cdot \frac{(x-x')}{\e}}dv; \displaybreak[0]\\
\widetilde{B}_7 &(x;x')\\
 = &- N\e \left[U_t(x) - U_t(x') - \nabla U_t\left(\dfrac{x+x'}{2}\right)\cdot(x-x') \right] \int (v\cdot \nabla_1 \widetilde{W}_{N,t}) \left(\dfrac{x+x'}{2},v\right) e^{i\,v\cdot \frac{(x-x')}{\e}}dv;
\end{align*}}
where we used the notation $\nabla_1$ and $\Delta_1$ to indicate derivatives with respect to the first variable.

In order to gain extra powers of $\e$, we write  
\begin{equation*}
\begin{split} 
&U_t (x) - U_t (x') - \nabla U_t \left( \frac{x+x'}{2} \right) \cdot (x-x') \\ 
&= \int_0^1 d\lambda \left[ \nabla U_t \left(\lambda x + (1-\lambda) x'\right) - \nabla U_t \left(\dfrac{(x+x')}{2}\right) \right] \cdot (x-x')  \\ 
&= \sum_{i,j=1}^3\int_0^1 d\lambda\left(\lambda-\frac{1}{2}\right) \int_0^1 d\mu\,\partial_i \partial_j U_t\left(\mu (\lambda x + (1-\lambda) x') + (1-\mu) \frac{(x+x')}{2}\right) (x-x')_i (x-x')_j.
\end{split} 
\end{equation*}
We notice that $U_t$ defined in \eqref{eq:convolution} has a convolution structure. Therefore   derivatives of $U_t$ are equivalent to derivatives of the spatial density $\widetilde{\varrho}_t$. Hence, when integrating out the $z$ variable in the Fefferman - de la Llave representation formula \eqref{eq:FDL-smooth}, we are left with
\be\label{eq:U-taylor}
\begin{split}
U_t &(x) - U_t (x') - \nabla U_t \left( \dfrac{x+x'}{2} \right) \cdot (x-x') \\
&= \sum_{i,j=1}^3\int_0^1 d\lambda\,\left(\lambda-\dfrac{1}{2}\right) \int_0^1 d\mu\,\int_0^\infty \dfrac{dr}{r^{2}}\\
&\ \times\int dy\,\chi_{(r,y)}\left(\mu (\lambda x + (1-\lambda) x') + (1-\mu)\dfrac{(x+x')}{2}\right)\partial_i\partial_j\widetilde{\varrho_t}(y)\, (x-x')_i (x-x')_j.
\end{split}
\ee
Inserting \eqref{eq:U-taylor} into the definition of $\widetilde{B}_1$, using twice the identity
\be\label{eq:v-der}
(x-x')\int\widetilde{W}_{N,t}\left(\frac{x+x'}{2},v\right)\,e^{iv\cdot\frac{x-x'}{\e}}\,dv=-i\e\int \nabla_v\widetilde{W}_{N,t}\left(\frac{x+x'}{2},v\right)\,e^{-iv\cdot\frac{x-x'}{\e}}dv 
\ee
and Young's inequality, we get
{\small
\begin{equation*}
\begin{split}
|\widetilde{B}_1&(x;x')|\\
\leq & C\,N\,\e^2\sum_{i,j=1}^3\int_0^1d\lambda\,\left|\lambda-\frac{1}{2}\right|\int_0^1 d\mu\left| \int_0^\infty\frac{dr}{r^{2}}\right.\\
\quad&\left.\int dy\,\chi_{(r,y)}(\mu(\lambda x+(1-\lambda)x')+(1-\mu)(x+x')/2)\partial^2_{i,j}\widetilde{\varrho}_t(y)\int dv\,\partial^2_{v_i,v_j}\widetilde{W}_{N,t}\left(\frac{x+x'}{2},v\right)e^{iv\cdot\frac{(x-x')}{\e}} \right|
\end{split}
\end{equation*}}

Therefore, the Hilbert-Schmidt norm of the operator $(1+x^2)\widetilde{B}_1$, where $(1+x^2)$ is the multiplication operator, can be estimated as follows:
\begin{equation*}
\begin{split}
\|(1+x^2)&\widetilde{B}_1\|_{\rm HS}^2\\
\leq& CN\e^4\int dq\int dp'\left[1+q^2+\e^2 p^2\right]^2\left|\sum_{i,j=1}^3\int_0^1d\lambda\,\left(\lambda-\frac{1}{2}\right)\int_0^1 d\mu \int_0^\infty\frac{dr}{r^{2}}\right.\\
&\quad\quad\quad\quad\left.\int dy\,\chi_{(r,y)}(q+\e\mu(\lambda-1/2)p)\partial^2_{v_i,v_j}\widetilde{\varrho}_t(y)\int dv\,\partial^2_{v_i,v_j}\widetilde{W}_{N,t}\left(q,v\right)e^{iv\cdot p} \right|^2
\end{split}
\end{equation*}
where we performed the change of variables 
\be\label{eq:change-var}
q=\dfrac{x+x'}{2}\,,\quad\quad\quad p=\dfrac{x-x'}{\e}
\ee
with Jacobian $J=8\,\e^3=8\,N$.\\
%
%
%
We fix $k>0$ and divide the integral into the two sets  $$A_{<}:=\{ r\in\R_+\ |\ r\leq k\} \quad\mbox{and}\quad A_>:=\{r\in\R_+\ |\ r>k\},$$ so that
\begin{equation}\label{eq:r-split}
\begin{split}
\|(1+&x^2)\widetilde{B}_1\|_{\rm HS}^2\\
\leq& CN\e^4\int dq\int dp[1+q^2+\e^2 p^2]^2\sum_{i,j=1}^3\int_0^1d\lambda\,\left|\lambda-\frac{1}{2}\right|\\
&\quad\int_0^1 d\mu \left|\int_{A_<}\frac{dr}{r^{2}}\int dy\,\chi_{(r,y)}(q+\e\mu(\lambda-1/2)p)\partial^2_{v_i,v_j}\widetilde{\varrho}_t(y)\int dv\,\partial^2_{v_i,v_j}\widetilde{W}_{N,t}\left(q,v\right)e^{iv\cdot p} \right|^2\\
+& CN\e^4\int dq\int dp[1+q^2+\e^2 p^2]^2\sum_{i,j=1}^3\int_0^1d\lambda\,\left|\lambda-\frac{1}{2}\right|\\
&\quad\int_0^1 d\mu \left|\int_{A_>}\frac{dr}{r^{2}}\int dy\,\chi_{(r,y)}(q+\e\mu(\lambda-1/2)p)\partial^2_{v_i,v_j}\widetilde{\varrho}_t(y)\int dv\,\partial^2_{v_i,v_j}\widetilde{W}_{N,t}\left(q,v\right)e^{iv\cdot p} \right|^2
\end{split}
\end{equation}
Denote by $\mathfrak{A}_<$ and $\mathfrak{A}_>$ the first and the second term of the sum on the r.h.s. of \eqref{eq:r-split} respectively. 
 For $\mathfrak{A}_<$ we use Young's inequality and H\"{o}lder's inequality with conjugated exponents $\theta=1$ and $\theta'=\infty$, then we perform the integral in the $y$ variable to extract $r^3$ which cancels the singularity, thus leading to the bound 
\be\label{eq:HS-B1-r0}
\mathfrak{A}_<\leq CN\e^4\int dq\int dv(1+q^2)^2|\nabla^2_v\widetilde{W}_{N,t}(q,v)|^2+CN\e^8\int dq\int dv\,|\nabla^4_v\widetilde{W}_{N,t}(q,v)|^2
\ee
where $C$ depends on $\|\nabla^2\widetilde{\varrho}_t\|_{L^\infty}$.\\
For $\mathfrak{A}_>$, we integrate by parts twice in the $y$ variable and recall that  $e^{-|z-y|^2/r^2}(1+|z-y|^2/r^2)$ is bounded uniformly in $z\in\R^3$. Since $\widetilde{\varrho}_t\in L^1(\R^3)$ we get the bound
\be\label{eq:HS-B1-rinf}
\mathfrak{A}_>\leq CN\e^4\int dq\int dv(1+q^2)^2|\nabla^2_v\widetilde{W}_{N,t}(q,v)|^2+CN\e^8\int dq\int dv\,|\nabla^4_v\widetilde{W}_{N,t}(q,v)|^2
\ee
where $C$ depends on $\|\widetilde{\varrho}_t\|_{L^1}$.\\
Whence, considering the two estimates \eqref{eq:HS-B1-r0}, \eqref{eq:HS-B1-rinf} together, we get
\be\label{eq:bound-HS-tildeB1}
\|(1+x^2)\widetilde{B}_1\|_{\rm HS}
\leq C\sqrt{N}\e^2\|\widetilde{W}_{N,t}\|_{H_2^2}
+C\sqrt{N}\e^4\|\widetilde{W}_{N,t}\|_{H^4}
\ee
where $C=C(\|\widetilde{\varrho}_t\|_{L^1},\|\nabla^2\widetilde{\varrho}_t\|_{L^\infty})$.

The Hilbert-Schmidt norms $\|(1+x^2)\widetilde{B}_3\|_{\rm HS}$, $\|(1+x^2)\widetilde{B}_4\|_{\rm HS}$ and $\|(1+x^2)\widetilde{B}_7\|_{\rm HS}$ can be handled analogously, thus obtaining 
\be
\|(1+x^2)\widetilde{B}_3\|_{\rm HS}\leq C\sqrt{N}\e^4\|\widetilde{W}_{N,t}\|_{H_4^4}
+C\sqrt{N}\e^6\|\widetilde{W}_{N,t}\|_{H_4^6}
\ee
\be
\|(1+x^2)\widetilde{B}_4\|_{\rm HS}\leq C\sqrt{N}\e^2\|\widetilde{W}_{N,t}\|_{H_4^2}
+C\sqrt{N}\e^4\|\widetilde{W}_{N,t}\|_{H_4^4}
\ee
\be
\|(1+x^2)\widetilde{B}_7\|_{\rm HS}\leq C\sqrt{N}\e^3\|\widetilde{W}_{N,t}\|_{H_3^2}
+C\sqrt{N}\e^5\|\widetilde{W}_{N,t}\|_{H_2^5}
\ee
To bound the  $\widetilde{B}_6$ term in which a higher order derivative of $U_t$ appears, we proceed as for $\widetilde{B}_1$: we first use \eqref{eq:U-taylor} and then divide the integral in the $r$ variables into two parts, according to the definition of the sets $A_<$ and $A_>$:
\be\label{eq:tildeB6}
\begin{split}
|\widetilde{B}_6&(x;x')|\\
\leq & C\,N\,\e^3\int_0^1d\lambda\,\left|\lambda-\dfrac{1}{2}\right|\int_0^1 d\mu\left| \int_{A_<}\dfrac{dr}{r^{2}}\int dy\,\nabla\chi_{(r,y)}\left((2\lambda-1)\dfrac{\mu}{2}(x-x')+\dfrac{(x+x')}{2}\right)\nabla^2\widetilde{\varrho}_t(y)\right.\\
&\quad\quad\quad\quad\quad\quad\quad\quad\quad\quad\quad\quad\quad\quad\quad\quad\quad\quad\quad\quad\quad\quad\quad\quad\quad\left.\int dv\,\nabla^2_{v}\widetilde{W}_{N,t}\left(\frac{x+x'}{2},v\right)e^{iv\cdot\frac{(x-x')}{\e}} \right|\\
+&C\,N\,\e^3\int_0^1d\lambda\,\left|\lambda-\dfrac{1}{2}\right|\int_0^1 d\mu\left| \int_{A_>}\dfrac{dr}{r^{2}}\int dy\,\nabla^3\chi_{(r,y)}\left((2\lambda-1)\dfrac{\mu}{2}(x-x')+\dfrac{(x+x')}{2}\right)\widetilde{\varrho}_t(y)\right.\\
&\quad\quad\quad\quad\quad\quad\quad\quad\quad\quad\quad\quad\quad\quad\quad\quad\quad\quad\quad\quad\quad\quad\quad\quad\quad\left.\int dv\,\nabla^2_{v}\widetilde{W}_{N,t}\left(\frac{x+x'}{2},v\right)e^{iv\cdot\frac{(x-x')}{\e}} \right|\\
\end{split}
\ee
where in the second term we have integrated by parts twice in the $y$ variable. 

We denote by $\widetilde{B}_6^<$ and $\widetilde{B}_6^>$ the operators with kernels defined respectively by the first and second term in the r.h.s. of \eqref{eq:tildeB6}. We consider $\|(1+x^2)\widetilde{B}_6^<\|_{\rm HS}$, perform the change of variables \eqref{eq:change-var} and choose $k$ such that $\int_0^k r^{1-\alpha}dr=1$, for some $\alpha\in(0,1)$. Then we can apply Young's inequality with measure  $r^{\alpha-1}dr$ and we get the bound
\be\label{eq:HS-B6-r0}
\begin{split}
\|(1+x^2)\widetilde{B}_6^<\|^2_{\rm HS}&\leq C\,N\,\e^6\int dq\int dp[1+q^2+\e^2p^2]^2\int_0^1d\lambda\,\left|\lambda-\dfrac{1}{2}\right|^2\int_0^1 d\mu \int_0^{k}\dfrac{dr}{r^{1-\alpha}}\frac{1}{r^{4+2\alpha}}\\
&\quad\quad\quad\quad\int dy\,\chi_{(r,y)}(q+\e\mu(\lambda-1/2)p)\frac{|q+\e\mu(\lambda-1/2)p-y|}{r}\\
&\quad\quad\quad\quad\int dy'\chi_{(r,y')}(q+\e\mu(\lambda-1/2)p)\frac{|q+\e\mu(\lambda-1/2)p-y'|}{r}\\
&\quad\quad\quad\quad\iint dv\,dv'\,\nabla^2_{v}\widetilde{W}_{N,t}(q,v)\,\nabla^2_{v'}\widetilde{W}_{N,t}(q,v')\,e^{i(v-v')\cdot p} \\
&\leq C\,N\,\e^6\|\widetilde{W}_{N,t}\|_{H^2_2}^2+C\,N\,\e^{10}\|\widetilde{W}_{N,t}\|_{H^4}^2
\end{split}
\ee
where $C$ depends on $\|\nabla^2\widetilde{\varrho}_t\|_{L^\infty}$.\\
For $r\in A_>$, we consider $\|(1+x^2)\widetilde{B}_6^>\|_{\rm HS}$. We perform the change of variables \eqref{eq:change-var} and recall that  $e^{-|z-y|^2/r^2}(|z-y|^k/r^k)\leq C$ for every $z\in\R^3$ and $k\in\N$. Since $\widetilde{\varrho}_s\in L^1(\R^3)$ we get the bound
\be\label{eq:HS-B6-rinf}
\|(1+x^2)\widetilde{B}_6^>\|_{\rm HS}^2\leq C\,N\,\e^6\|\widetilde{W}_{N,t}\|_{H^2_2}^2+C\,N\,\e^{10}\|\widetilde{W}_{N,t}\|_{H^4}^2
\ee
where $C$ depends on $\|\widetilde{\varrho}_t\|_{L^1}$.\\

Thus considering the two terms together we get the desired bound 
\be\label{eq:bound-HS-tildeB6}
\|(1+x^2)\widetilde{B}_6\|_{\rm HS}
\leq C\sqrt{N}\e^3\|\widetilde{W}_{N,t}\|_{H_2^2}
+C\sqrt{N}\e^5\|\widetilde{W}_{N,t}\|_{H^4}
\ee
where $C=C(\|\widetilde{\varrho}_t\|_{L^1},\|\nabla^2\widetilde{\varrho}_t\|_{L^\infty})$.

The norms $\|(1+x^2)\widetilde{B}_j\|_{\rm HS}$, $j=2,5$, can be dealt analogously, thus obtaining
\be\label{eq:bound-HS-tildeB6}
\|(1+x^2)\widetilde{B}_2\|_{\rm HS}
\leq C\sqrt{N}\e^4\|\widetilde{W}_{N,t}\|_{H_2^2}
+C\sqrt{N}\e^6\|\widetilde{W}_{N,t}\|_{H^4}
\ee
and 
\be\label{eq:bound-HS-tildeB6}
\|(1+x^2)\widetilde{B}_5\|_{\rm HS}
\leq C\sqrt{N}\e^4\|\widetilde{W}_{N,t}\|_{H_2^4}
+C\sqrt{N}\e^6\|\widetilde{W}_{N,t}\|_{H^6}
\ee
where $C=C(\|\widetilde{\varrho}_t\|_{L^1},\|\nabla^2\widetilde{\varrho}_t\|_{L^\infty})$.

Gathering together all the terms, we get
\be
\begin{split} 
\| (1 + x^{2}) &\widetilde{B} \|_\text{HS}\\
& \leq C \sqrt{N} \left[ \eps^2 \| \widetilde{W}_{N,t} \|_{H^2_{4}} + \eps^3  \| \widetilde{W}_{N,t} \|_{H^3_{4}} + \eps^4 \| \widetilde{W}_{N,t} \|_{H^4_{4}} + \eps^5 \| \widetilde{W}_{N,t} \|_{H^5_4}+\eps^6 \| \widetilde{W}_{N,t} \|_{H^6_4} \right] 
\end{split}
\ee

\end{proof}

\section{Proof of Theorem \ref{thm:tr-HS}}\label{sect:proof}


\noindent From Lemma \ref{lemma:traceH-V} Eq. \eqref{eq:traceH-V} and Proposition \ref{prop:error-term}, we know that 
\be\label{eq:trace-dom+err}
\tr\ |\omega_{N,t}-\widetilde{\omega}_{N,t}|\leq \frac{1}{\e}\int_0^t\tr\left|\left[\frac{1}{|\cdot|}*(\varrho_s-\widetilde{\varrho}_s)\,,\,\widetilde{\omega}_{N,s}\right]\right|\,ds
+ C\,N\,\e\,\int_0^t \left(1+\sum_{k=1}^4 \e^k\|\widetilde{W}_{N,s}\|_{H_4^{k+2}}\right)\,ds.
\ee
We focus on the first term on the r.h.s. of \eqref{eq:trace-dom+err}. Recalling Lemma \ref{lemma:dominant} Eq. \eqref{eq:dom-term}, we fix a positive real number $k$ and we write
\be\label{eq:trace-est}
\begin{split}
\tr\left|\left[\frac{1}{|\cdot|}*(\varrho_s-\widetilde{\varrho}_s)\,,\,\widetilde{\omega}_{N,s}\right]\right|&\leq C\int_0^\infty \frac{1}{r^{5}}\iint |\varrho_s(y)-\widetilde{\varrho}_s(y)|\,\chi_{(r,z)}(y)\,\tr\ |[\chi_{(r,z)}\,,\,\widetilde{\omega}_{N,s}]|\,dz\,dy\,dr\\
&\leq C\int_0^k \dfrac{1}{r^{\frac{7}{2}+3\delta}}\int |\varrho_s(y)-\widetilde{\varrho}_s(y)|\,g_r(y)\,\sum_{i=1}^3\|\varrho_{|[x_i,\widetilde{\omega}_{N,s}]|}\|_{L^1}^{\frac{1}{6}+\delta}\,dy\,dr\\
&+C\sum_{i=1}^3\|\varrho_{|[x_i,\widetilde{\omega}_{N,s}]|}\|_{L^1}\int_k^\infty\frac{1}{r^6}\int |\varrho_s(y)-\widetilde{\varrho}_s(y)|\,dy\,dr
\end{split}
\ee
where 
\begin{equation*}
g_r(y)=\int \chi_{(r,z)}(y)\,\left(\varrho_{|[x,\widetilde{\omega}_{N,s}]|}^*(z)\right)^{\frac{5}{6}-\delta}\,dz\,,
\end{equation*}
and we used Lemma \ref{lemma:tr1} in the second line of equation \eqref{eq:trace-est} and the bound \eqref{eq:tr1-L1} in the last line of equation \eqref{eq:trace-est}.

We now compute the $L^\infty$ norm of $g_r(y)$:
\begin{equation*}
\|g_r\|_{L^\infty}\leq Cr^{\frac{3}{q}}\|\varrho^*_{|[x,\widetilde{\omega}_{N,s}]|}\|_{L^{(\frac{5}{6}-\delta)q'}}^{\frac{5}{6}-\delta}\leq Cr^{\frac{3}{q}}\|\varrho_{|[x,\widetilde{\omega}_{N,s}]|}\|_{L^{(\frac{5}{6}-\delta)q'}}^{\frac{5}{6}-\delta}
\end{equation*}
where $q$ and $q'$ are conjugated H\"{o}lder exponents and we have used the $L^s$ boundedness of the Hardy-Littlewood maximal operator in the last inequality, for $s=(\frac{5}{6}-\delta)q'>1$. 

To deal with the singularity at zero in the $r$ variable in \eqref{eq:trace-est}, we choose $q'>6$ and $q<6/5$.
%
Hence, there exist a constant $C_{t,1}$, depending on time but independent on $N$, such that
\be\label{eq:dominated}
\begin{split}
\tr\ \left|\left[\frac{1}{|\cdot|}*(\varrho_s-\widetilde{\varrho}_s)\,,\,\widetilde{\omega}_{N,s}\right]\right|&\leq C\|\varrho_s-\widetilde{\varrho}_s\|_{L^1}\,\sum_{i=1}^3\left(\|\varrho_{|[x_i,\widetilde{\omega}_{N,s}]|}\|_{L^1}^{\frac{1}{6}+\delta}\|\varrho_{|[x_i,\widetilde{\omega}_{N,s}]|}\|_{L^p}^{\frac{5}{6}-\delta}+\|\varrho_{|[x_i,\widetilde{\omega}_{N,s}]|}\|_{L^1}\right)\\
&\leq C_{t,1}\,\e\,\tr\ |\omega_{N,s}-\widetilde{\omega}_{N,s}|\,,
\end{split}
\ee
where in the last inequality we used assumption \eqref{eq:sc-assumptions} with $p>5$.

We now analyse the second term on the r.h.s. of \eqref{eq:trace-dom+err}. Using assumptions $i),\,ii)$ and $iii)$ in Theorem \ref{thm:tr-HS} and a trivial adaptation of Appendix A in \cite{S19}, we can bound the weighted Sobolev norms $\|\widetilde{W}_{N,t}\|_{H^{k+2}_4}$, for $k=1,\dots,4$, in terms of the initial data $W_N$:
\be\label{eq:errors-0}
\|\widetilde{W}_{N,t}\|_{H_4^{k+2}}\leq C_t\|{W}_N\|_{H_4^{k+2}}\,,
\ee
where $C_t$ is a time dependent constant, for $t\in[0,T]$.  

Therefore, Eq. \eqref{eq:dominated} and Eq. \eqref{eq:errors-0} leads to the Gr\"{o}nwall type estimate
\begin{equation*}
\tr\ |\omega_{N,t}-\widetilde{\omega}_{N,t}|\leq {C}_{t,1}\,\int_0^t \tr\ |\omega_{N,s}-\tilde{\omega}_{N,s}|\,ds
+ C_{t,2}\,N\,\e\,\int_0^t \left(1+\sum_{k=1}^4 \e^k\|{W}_{N}\|_{H_4^{k+2}}\right)\,ds\,,
\end{equation*} 
where, for every fixed $T>0$ and for all $t\in[0,T]$, ${C}_{t,1}$ is proportional to the constant appearing in assumption \eqref{eq:sc-assumptions} and $C_{t,2}$ depends on $t\in[0,T]$ and on $\|W_N\|_{H_4^2}$. Both $C_{t,1}$ and $C_{t,2}$ are independent of $N$. Hence, by Gr\"{o}nwall Lemma, we conclude the proof of \eqref{eq:tr-norm}.
\medskip

\section{Translation invariant and steady states of the Vlasov-Poisson system}\label{sect:steady-states}

In general, for $T>0$ fixed, we do not know which hypotheses $\omega_N$ should satisfy at time $t=0$ in order for the bounds \eqref{eq:sc-assumptions} to hold for all $t\in[0,T]$.  In this section, we are interested in identifying a special class of states which satisfy the assumptions of Theorem \ref{thm:tr-HS}. 

We start by considering a sequence of fermionic reduced densities $\omega_N$ which are superpositions of coherent states 
\begin{equation*}
f_{q,p}(x)=\e^{-\frac{3}{2}}e^{-ip\cdot x/\e}G(x-q)
\end{equation*}
where, for every $\delta>0$, $G$ is the Gaussian defined as follows
\begin{equation*}
G(x-q)=\dfrac{e^{-{|x-q|^2}/{2\delta^2}}}{(2\pi\delta^2)^{\frac{3}{4}}}\,.
\end{equation*}
Namely, for $M:\R^3\times\R^3\to\R_+$ probability density such that $0\leq M(q,p)\leq 1$ for all $(q,p)\in\R^3\times\R^3$ and $\iint M(q,p)\,dq\,dp=1$, we define the sequence of fermionic operators
\be\label{eq:super-coherent}
\omega_N=\iint M(q,p)\,|f_{(q,p)}\rangle\langle f_{(q,p)}|\,dq\,dp\,,
\ee
with kernel
\be\label{eq:kernel-super-coherent}
\omega_{N}(x;y)=\iint_{\R^3\times\R^3} M(q,p)\,f_{(q,p)}(x)\,\overline{f_{(q,p)}(y)}\,dq\,dp\,, 
\ee
where  in formula \eqref{eq:kernel-super-coherent} we used the bra-ket notation.

We notice that, if $M\in\mathcal{W}^{1,1}(\R^3\times\R^3)$, the Sobolev space of functions such that the norm $\|\nabla M\|_{L^1}$ is finite, then the sequence $\omega_N$ defined as in \eqref{eq:super-coherent} satisfies the bound
\begin{equation*}
\|\varrho_{|[x,\omega_N]|}\|_{L^1(\R^3)}=\tr\ |[x,\omega_N]|\leq CN\e\,,
\end{equation*}
where $C=\|\nabla_v M\|_{L^1(\R^3_x\times\R^3_v)}$. Indeed, consider the kernel of the commutator $[x,\omega_N]$:
\begin{equation*}
\begin{split}
[x,\omega_N](x;y)&=\iint (x-y)\,M(q,p)\,f_{(q,p)}(x)\,\overline{f_{(q,p)}(y)}\,dq\,dp\\
&=\dfrac{N\e}{i}\iint \nabla_p M(q,p)\,f_{(q,p)}(x)\,\overline{f_{(q,p)}(y)}\,dq\,dp
\end{split}
\end{equation*}
thus
\begin{equation*}
[x,\omega_N]=\dfrac{N\e}{i}\iint \nabla_p M(q,p)\,|f_{(q,p)}\rangle\langle f_{(q,p)}|\,dq\,dp
\end{equation*}
hence the trace norm is easily bounded as follows
\begin{equation*}
\tr\ |[x,\omega_N]|\leq N\e\iint |\nabla_p M(q,p)|\,\|f_{(q,p)}\|_{L^2(\R^3)}^2\,dq\,dp=N\e\|\nabla_v M\|_{L^1(\R^3\times\R^3)}\,.
\end{equation*}
We do expect that, under suitable regularity and integrability assumptions on $M$, there exists a finite positive constant $C$ such that $\|\varrho_{|[x,\omega_N]|}\|_{L^p(\R^3)}\leq CN\e$ for some $p>5$. 

Now consider $\omega_N$ defined in \eqref{eq:super-coherent} and satisfying the bound
\begin{equation}\label{eq:bounds-0}
\|\varrho_{|[x,\omega_N]|}\|_{L^1}+\|\varrho_{|[x,\omega_N]|}\|_{L^p} \leq C\,N\e\,,\quad\quad p>5
\end{equation}
to be the initial datum of the Cauchy problem associated with Eq. \eqref{eq:VP-weyl}. Its evolution is denoted by $\widetilde{\omega}_{N,t}$. 

In particular, if we choose $\tilde\omega_{N,t}$ to be translation invariant, no matter whether the interaction is attractive or repulsive, $\|\varrho_{|[x,\widetilde{\omega}_{N,t}]|}\|_{L^p}=\|\varrho_{|[x,\widetilde{\omega}_{N,0}]|}\|_{L^p}=\|\varrho_{|[x,{\omega}_{N}]|}\|_{L^p}$.

If we restrict to attractive interactions, then the existence of steady states for the Vlasov-Poisson system (cf. \cite{Rein}) allows to enlarge the class of initial data. Indeed, if $\widetilde{\omega}_{N,t}$ is a steady state for the Weyl transformed Vlasov-Poisson system then, for every fixed $T>0$, $\widetilde{\omega}_{N,t}$ automatically satisfies the bound
\begin{equation}\label{eq:bounds-t}
\|\varrho_{|[x,\widetilde{\omega}_{N,t}]|}\|_{L^1}+\|\varrho_{|[x,\widetilde{\omega}_{N,t}]|}\|_{L^p} \leq C\, N\e\,,\quad\quad p>5
\end{equation}
for all $t\in (0,T]$, if it does at time $t=0$. 

Moreover, if $\widetilde{\omega}_{N,t}$ is a steady state for Eq. \eqref{eq:VP-weyl} with gravitational interaction, then its Wigner transform $\widetilde{W}_{N,t}$ solves the equation
\be\label{eq:VP-steady}
\left\{\begin{array}{l}
v\cdot\nabla_x\widetilde{W}_{N,t}- \nabla_x U\cdot\nabla_v\widetilde{W}_{N,t}=0\,,\\\\
-\Delta_xU(t,x)=\widetilde{\varrho}_t(x)\,,\\\\
\widetilde{\varrho}_t(x)=\int \widetilde{W}_{N,t}(x,v)\,dv\,.
\end{array}
\right.
\ee

One example of states which satisfy \eqref{eq:VP-steady} are functions of the form
\begin{equation*}
\widetilde{W}_{N,t}(x,v)=\Phi\circ H(t,x,v)
\end{equation*}
where $\Phi$ is a smooth function of the local energy
\be\label{eq:energy-H}
H(t,x,v)=\dfrac{|v|^2}{2}- U(t,x)\,.
\ee
In particular, choose $\omega_N$ as in definition \eqref{eq:super-coherent} and assume $M$ to be such that $M*G$, where $G$ is a normalized Gaussian. Then we can express it as the composition of a smooth compactly supported function $\Phi$ with the regularized local energy 
\be\label{eq:energy-H-smooth}
\widetilde{H}(t,x,v)=\dfrac{|v|^2}{2}-(U*G)(t,x)\,.
\ee
With this choice, $\omega_N$ is a fermionic operator and $M*G=\Phi\circ \widetilde{H}$ satisfies Eq. \eqref{eq:VP-steady}, that is $\widetilde{\omega}_{N,t}$ is a steady state for \eqref{eq:VP-weyl}. Indeed, let us denote $\widetilde{M}_t(q,p)=M_t(\cdot,p)*G(q)$ and consider
{\small
\begin{equation*}
v\cdot\nabla_x\widetilde{W}_{N,t}(x,v)=\dfrac{1}{(2\pi)^3}\iiint M_t(q,p)e^{-i(p+v)\cdot y}\dfrac{2}{\delta^2}v\cdot(x-q)\,G(x-q+\e y/2)\,G(x-q-\e y/2)\,dq\,dp\,dy
\end{equation*}}
and 
{\small
\begin{equation*}
\begin{split}
&\nabla_x U(t,x)\cdot\nabla_v\widetilde{W}_{N,t}(x,v)\\
&=-\dfrac{1}{(2\pi)^3}\iiint\nabla_p M_t(q,p)\cdot\left[\nabla\dfrac{1}{|\cdot|}*\left(\varrho_{M_t}*G^2\right)(x)\right]G\left(x-q+\e \dfrac{y}{2}\right)G\left(x-q-\e \dfrac{y}{2}\right)e^{-(v+p)\cdot y}\,dq\,dp\,dy\,
\end{split}
\end{equation*}}
where $\varrho_{M_t}(x)=\int M_t(x,v)\,dv$. \\ We insert the two above expressions into \eqref{eq:VP-steady} and we obtain
{\small
\begin{equation*}
\begin{split}
&v\cdot\nabla_x\widetilde{W}_{N,t}- \nabla_x U\cdot\nabla_v\widetilde{W}_{N,t}\\
&=\frac{1}{(2\pi)^3}\iint \left[ -v\cdot\nabla_x\widetilde{M}_t(x,p)-\nabla_x\left(|\cdot|^{-1}*(\varrho_{M_t}*G^2)\right)(x)\cdot\nabla_p\widetilde{M}_t(x,p)\right]G\left(\frac{\e y}{4}\right)e^{-i(p+v)\cdot y}\,dy\,dp\,.
\end{split}
\end{equation*}}
The choice $\widetilde{M}_t=\Phi\circ \widetilde{H}$ guarantees the r.h.s. of the above expression to be equal to zero. Therefore $\widetilde{\omega}_{N,t}$ is a steady state of the Weyl transformed Vlasov-Poisson system \eqref{eq:VP-weyl}. Moreover, if we assume \eqref{eq:bounds-0}, the bounds \eqref{eq:bounds-t} are satisfied.

\medskip

{\bf Acknowledgement.} The support of the Swiss National Science Foundation through the Ambizione grant ${\rm PZ00P2\_161287 / 1}$ and the Eccellenza fellowship {\rm PCEFP2\_181153} is gratefully acknowledged.


\adresse

 \end{document}